\documentclass[11pt,a4paper]{article}
\pdfoutput=1 

\usepackage[margin=2.4cm]{geometry}
\usepackage{amsmath,amssymb,amsfonts,amsthm}
\usepackage{caption}
\usepackage{graphicx}
\usepackage{booktabs}
\usepackage{placeins}
\usepackage{xcolor}
\usepackage[numbers,sort&compress]{natbib}
\usepackage[colorlinks=true,linkcolor=blue!60!black,citecolor=blue!60!black,urlcolor=blue!60!black]{hyperref}

\newcommand{\C}{\mathbb{C}}
\newcommand{\R}{\mathbb{R}}

\newcommand{\Ct}{\mathcal{C}}
\newcommand{\ov}[1]{\overline{#1}}
\newtheorem{theorem}{Theorem}
\newtheorem{definition}[theorem]{Definition}

\newtheorem{proposition}[theorem]{Proposition}
\newtheorem{corollary}[theorem]{Corollary}

\title{\bfseries Realified tensor networks: quantum circuit simulation\\ on real-valued matrix accelerators}
\author{Yusheng Zhao\textsuperscript{1,2},
  Xiwei Pan\textsuperscript{1},
  Enji Xiong\textsuperscript{2,3},
  Chengkai Zhu\textsuperscript{2},
  Jin-Guo Liu\textsuperscript{1,2,$\star$} \\[2pt]
  \small $^{1}$Thrust of Advanced Materials, The Hong Kong University of\\
  \small Science and Technology (Guangzhou), Guangdong 511453, China\\
  \small $^{2}$QudeLeap Research, Shanghai 200030, China\\
  \small $^{3}$Thrust of Artificial Intelligence, The Hong Kong University of\\
  \small Science and Technology (Guangzhou), Guangdong 511453, China\\[2pt]
  \small $\star$ \href{mailto:jinguoliu@hkust-gz.edu.cn}{jinguoliu@hkust-gz.edu.cn}}
\date{\small\itshape August 4, 2026}

\begin{document}
\maketitle

\begin{abstract}
Tensor-network contraction simulates quantum circuits, but modern matrix
accelerators (NPUs, TPUs) expose only real GEMM pipelines, so the complex
networks of quantum simulation must be reconstructed in software. We resolve
the mismatch by a realification rewrite that maps any complex tensor network to
a real one. At each merge of two complex tensors, a rank-3 structure tensor
realizes Gauss's three-multiplication (3M) formula; contractions with one or no
complex operand need only two or one real products. We prove a tight cost law:
overhead $1 + 2m + r$ in real multiplications, where $m$ and $r$ are the volume
fractions of two- and one-complex-operand contractions, never exceeding
$3\times$ relative to real contraction, with every intermediate at most
doubled in size. On 67 circuits (random, Clifford+$T$, QAOA, VQE), the law holds
across the real-to-complex range and complex-gate placement, not count,
governs cost. Contraction orders transfer from the complex network with a
relative arithmetic-cost gap below $5\times 10^{-4}$ on 66 of 67 circuits; the
exception closes under a few steps of low-temperature simulated annealing. On an
Ascend 910 NPU the rewrite beat both the four-real-GEMM baseline and a per-GEMM
Gauss lowering on all twelve random circuits and on 52 of 55 structured cells
(three cells slower by at most 12\%); the four-GEMM baseline was slower by a
median $1.7\times$ (random) and $1.4\times$ (structured). Realification makes
complex tensor-network contraction native to real-only matrix engines.
\end{abstract}

\section{Introduction}

\begin{figure}[tb]
  \centering
  \includegraphics[width=0.98\linewidth]{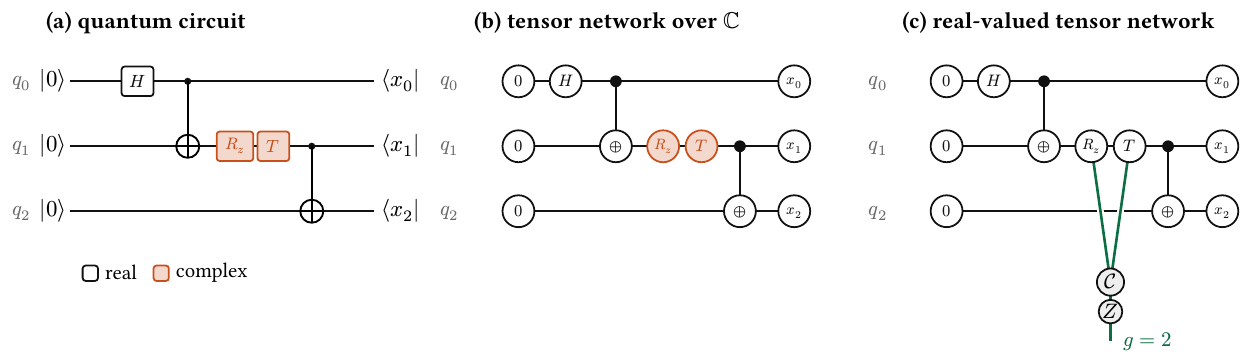}
  \caption{From quantum circuit to real-valued contraction. (a)~A quantum circuit $U$
  with real (white) and complex (orange) gates; the simulation target is a single
  amplitude $\langle x|U|0\rangle$, i.e.\ the projection of the output onto the
  computational-basis product state $|x\rangle = |x_0 x_1 x_2\rangle$. (b)~Its
  tensor-network representation over $\C$ (schematic): gates become tensors and qubit
  lines become contracted indices. (c)~Realification as a local graph rewrite: each complex tensor
  acquires a dimension-2 leg (green) carrying its real and imaginary parts,
  contracted green legs meet at the structure tensor $\Ct$, and a conjugated
  tensor carries the sign flip $Z = \mathrm{diag}(1,-1)$ on its green leg (Sec.~\ref{sec:rep});
  the result is a single static real einsum.}
  \label{fig:overview}
\end{figure}

Tensor-network contraction simulates quantum
circuits at scale, from single-amplitude evaluation \cite{MarkovShi2008,Villalonga2019} to
the verification \cite{Villalonga2019,GrayKourtis2021} and spoofing \cite{Pan2022} of
random-circuit-sampling experiments \cite{Arute2019}. Its computational core is a schedule
of large matrix multiplications over the complex numbers. Hardware, however,
has moved in the opposite direction: the matrix engines of modern accelerators, including
TPU matrix units \cite{Lewis2022}, GPU tensor cores \cite{Ootomo2023}, and NPUs such as
Huawei Ascend, expose only real floating-point types on their matrix-multiply paths.
Complex multiplication must be reconstructed in software.

Existing practice reconstructs complex multiplication at the general matrix multiply (GEMM) level. The textbook expansion
$(a + ib)(c + id) = (ac - bd) + i(ad + bc)$ costs four real multiplications per complex
multiplication, hence four real GEMMs per complex GEMM, hereafter \emph{GEMM-4M},
the textbook lowering; on real-only accelerators such as
TPUs, complex tensor arithmetic is likewise reconstructed from real-valued
GEMM primitives~\cite{Hauru2021,Morningstar2022,Ganahl2023}. Gauss's
three-multiplication method, hereafter \emph{GEMM-3M}, trades one real GEMM for
three additional matrix additions (Eq.~\eqref{eq:gauss3m}). Its multiplication count is
optimal \cite{Winograd1971}; although less stable than GEMM-4M, it is accurate enough
for practical use \cite{Higham1992}. Both
lowerings underlie high-performance complex-GEMM kernels~\cite{VanZee2017}. Either lowering, applied to every contraction as fully complex, pays a flat $4\times$ or $3\times$; a mostly-real circuit gains nothing from its real gates.
In contrast, \emph{circuit-level} realification theory
\cite{Rudolph2002,Shi2003,Aharonov2003,McKague2009} compiles any complex quantum circuit into a real
one by adding a single ancilla qubit whose state distinguishes the real and imaginary
components of each amplitude. Each complex gate is replaced by a real gate that also acts
on the ancilla. So the ancilla qubit line, a bond of dimension~2, runs through the full
circuit in gate order: the \emph{worldline}. The worldline is a physical bond: the ancilla is generically entangled with the
data register. Each worldline edge that an intermediate's subtree cuts leaves an open
dimension-2 index on that intermediate, so cutting $k$ edges inflates it by $2^k$
(Sec.~\ref{sec:law}, remarks following Theorem~\ref{thm:law}). Keeping the cut
count bounded forces every subtree to be a union of a few contiguous windows of
gate order, i.e.\ transfer-matrix-style contraction, whose intermediates grow
exponentially with circuit width. Geometry-driven orders~\cite{MarkovShi2008,GrayKourtis2021}, the ones optimizers find
worth using, scatter each subtree across time and generically cut the worldline at
many points. The frozen wiring, set by the circuit's time structure rather than the
network's geometry, thus exacts a trade-off between a poor contraction order and
exponential memory inflation.

We take a third route: realify the tensor network itself by a local graph rewrite, hereafter \emph{network-3M}
(Fig.~\ref{fig:overview}).
The underlying algebraic primitive, viewing $\C$ as a 2-dimensional real
*-algebra whose structure tensor mediates
multiplication, appears in *-tensor formulations of quantum
mechanics~\cite{Bauer2020,BauerNietner2022}; we develop it into a
contraction-level method (Sec.~\ref{sec:rep}).
The rewrite turns the complex network into an equivalent real one,
distributing the real--imaginary bookkeeping along the contraction tree rather than a
single worldline, so the contraction order remains a free choice.
\begin{enumerate}
\item We map any complex tensor
network to a real one purely by graph rewriting (Sec.~\ref{sec:rep}): each complex tensor
acquires one dimension-2 auxiliary index, complex multiplication enters through a single
permutation-symmetric structure tensor $\Ct$, and conjugation and global phases become
local linear operators. The map itself is the *-algebra realification
of~\cite{Bauer2020,BauerNietner2022}; the new element is its compilation: the $\Ct$
wiring is aligned with a chosen contraction tree (Sec.~\ref{sec:algebra}), so the
rewrite is the entire intervention. The realified network is
an ordinary real einsum with fixed shapes, so the existing software stack is reused
without modification: real GEMM kernels on the accelerator, tensor-network contraction
and order optimization, and static-graph compilers~\cite{Sabne2020,Chen2018,Ansel2024}
all operate on it.
Order optimization runs on the
complex network before conversion and the orders transfer (Sec.~\ref{sec:bench}).
Optimizers that price the realified factor graph directly use a different
loop-volume convention (counting loop iterations rather than scalar
multiplications), quantified in Appendix~\ref{app:indopt}.

\item We prove a tight cost law (Sec.~\ref{sec:law}). The arithmetic overhead
relative to the \emph{real skeleton} (the same network contracted over $\R$)
is $1 + 2m + r$, where $m$ and $r$ are the volume fractions of steps
with two or one complex operand (merges and rides). Every intermediate at
most doubles in size. The bound lies between
$1\times$ in the all-real limit and $3\times$ in the all-complex limit, so realification
is never more expensive than 3M emulation and automatically cheaper when the circuit has
real structure, an advantage that flat per-GEMM lowering cannot exploit (a per-GEMM
dispatcher with operand-realness checks matches the arithmetic count but not the
structural benefits of Sec.~\ref{sec:conclusions}).

\item We validate the method in three ways. First,
contraction-volume audits across 67 circuits (twelve random circuits from the qflex
collection~\cite{Villalonga2019} plus 55 Clifford+$T$, QAOA, and VQE instances
spanning the real-to-complex range) verify the law's accounting and locate
each circuit on its range, from the $1\times$ all-real
limit to the $3\times$ ceiling set by the real rank of complex
multiplication~\cite{Winograd1971} (Tables~\ref{tab:bench}
and~\ref{tab:struct}). Complex-gate \emph{placement} rather than count governs
the cost: at fixed $T$-count, placement alone moves the overhead from $1.98$
to $2.93$.
Second, orders optimized on the complex network survive realification
(Fig.~\ref{fig:pipe}): convert-only matches full
reoptimization to within $5\times10^{-4}$ on 66 of 67 circuits, and the sole
exception, a 5-qubit test circuit, is closed by a short low-temperature
polish, so the pipeline needs no full reoptimization pass. A landscape theory explains and delimits this
flatness (Appendix~\ref{app:landscape}). Step counts are tree
invariants, and a density condition on complex leaves pins every cheap
tree near the $3\times$ ceiling. Reusing the complex-optimal tree costs
at most $3\times$ the realified optimum on any network, yet can approach
$2\times$ on an explicit green-sparse family: flatness is a property of
the circuit class, not of the method.
Third, on Ascend 910, the realified executor was faster than both GEMM-4M and
GEMM-3M on all twelve random circuits and on 52 of 55 structured
\emph{device-clean} cells, i.e.\ cells free of a documented device-software
fallback; the three non-wins lie at $0.90$--$0.99{\times}$. Median speedups over the 4M
baseline were ${\approx}\,1.7\times$ (random) and ${\approx}\,1.4\times$ (structured)
(Sec.~\ref{sec:wallclock}); the arithmetic saving
translates to end-to-end speedup only when the factorization is embedded at the
network level.
\end{enumerate}

\noindent Section~\ref{sec:rep} develops the representation and proves the cost law,
and Sec.~\ref{sec:bench} reports the benchmarks and hardware validation.

\section{Realified tensor networks}\label{sec:rep}

The realification of individual contractions via the structure tensor of the $*$-algebra $\C/\R$ was introduced by Bauer~\cite{Bauer2020,BauerNietner2022}; here we compile it into a full contraction-level method with provable cost bounds and hardware validation.
Sections~\ref{sec:rep:local} and~\ref{sec:algebra} present the
realification as a self-contained algebraic theory: a single complex
contraction is realified, then the construction is shown to be independent of contraction order.
Section~\ref{sec:law} turns to the practical consequences, deriving the
arithmetic cost law and showing that reverse-mode differentiation
requires no complex-number extensions.

\subsection{The representation and the structure tensor}\label{sec:rep:local}

\begin{definition}[Realification]\label{def:realify}
Let $A$ be a complex tensor, decomposed as $A = A_R + i A_I$ with $A_R, A_I$ real. Its
\emph{realification} $T_A$ is the real tensor obtained by stacking $A_R$ and $A_I$ along one
auxiliary index $g$ of dimension~2. A tensor is
\emph{green} if it carries such an index; real tensors need none.
\begin{center}
  \includegraphics[height=2.6cm]{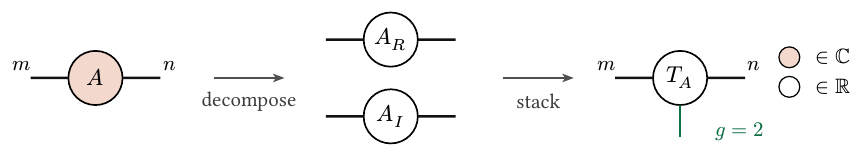}%
\\[2pt]{\small Realification: a complex tensor $A = A_R + iA_I$ becomes a
real tensor $T_A$ with one auxiliary dimension-2 index.}
\end{center}
\end{definition}

\paragraph{The multiplication tensor.}
Realification turns each tensor real, but a pairwise contraction must still
reproduce the complex product. The product is recovered by inserting a
\emph{multiplication tensor} $M$ between the green legs of the two operands.

\begin{definition}[Multiplication tensor]\label{def:mult}
The \emph{multiplication tensor} $M \in \R^{2 \times 2 \times 2}$ encodes
the product of $\C$ as a real algebra:
\begin{equation}\label{eq:mult}
  (xy)_c = \sum_{a,b=1}^{2} M_{abc}\, x_a y_b, \qquad
  M_{\cdot\cdot 1} = \begin{pmatrix} 1 & 0\\ 0 & -1 \end{pmatrix}, \qquad
  M_{\cdot\cdot 2} = \begin{pmatrix} 0 & 1\\ 1 & 0 \end{pmatrix},
\end{equation}
with indices $a, b, c \in \{1, 2\}$ ($1$~real, $2$~imaginary).
\end{definition}

\noindent Contracting the green legs of two realified tensors with $M$ reproduces
the complex contraction exactly: each complex product becomes one application
of $M$, and addition is componentwise.

\paragraph{Gauss's 3M lowering.}
For complex matrices $A=A_R+iA_I$ and $B=B_R+iB_I$, Gauss's method forms
\begin{equation}\label{eq:gauss3m}
  \begin{aligned}
    P_1&=A_RB_R, & P_2&=A_IB_I,
    & P_3&=(A_R+A_I)(B_R+B_I),\\
    AB&=(P_1-P_2)+i(P_3-P_1-P_2).
  \end{aligned}
\end{equation}
Thus GEMM-3M uses three real matrix multiplications and five real matrix additions,
whereas GEMM-4M uses four and two, respectively~\cite{Higham1992,VanZee2017}.
Equation~\eqref{eq:rank3}, displayed with the proof of Theorem~\ref{thm:law}, is the tensor form of the same bilinear identity, with
the output sign flip separated as $M=\Ct Z$ (Definition~\ref{def:structure}). Network-3M embeds this rank-3
factorization into the graph only where two green subgraphs merge.
\begin{center}
  \includegraphics[height=2.6cm]{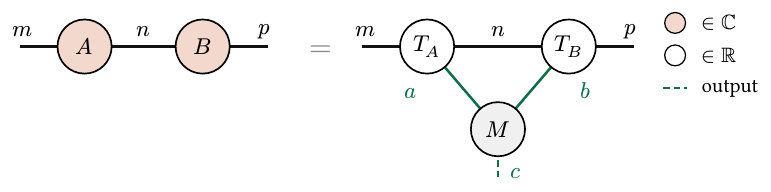}%
\\[2pt]{\small Multiplication via $M$: contracting the green legs of
$T_A$ and $T_B$ through $M$ reproduces the complex product.}
\end{center}

\paragraph{Symmetrization.}
$M$ distinguishes its legs: the first two carry inputs, the third carries the
output (Fig.~\ref{fig:sym}(a)). In a tensor network, nodes are wired in
arbitrary orientations, so we need a fully symmetric replacement.

\begin{definition}[Structure tensor]\label{def:structure}
Let $Z = \mathrm{diag}(1,-1)$. The \emph{structure tensor} $\Ct$ is obtained
by absorbing $Z$ into the output leg of $M$:
\begin{equation}\label{eq:ctensor}
  \Ct_{abc} = \sum_{c'} M_{abc'} Z_{c'c}
  = \operatorname{Re}\!\left(i^{\,a+b+c-3}\right), \qquad
  \Ct_{\cdot\cdot 1} = \begin{pmatrix} 1 & 0\\ 0 & -1 \end{pmatrix}, \qquad
  \Ct_{\cdot\cdot 2} = \begin{pmatrix} 0 & -1\\ -1 & 0 \end{pmatrix}.
\end{equation}
The closed form depends only on the sum $a+b+c$: $\Ct$ is invariant under
all six permutations of its legs (Fig.~\ref{fig:sym}(b)).
This tensor coincides with the complex-number *-algebra tensor
of Bauer~\cite{Bauer2020}.
\end{definition}

\begin{figure}[htb]
  \centering
  \includegraphics[width=0.56\linewidth]{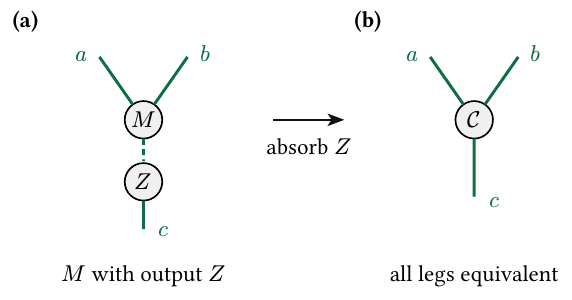}
  \caption{(a)~$M$ distinguishes input legs (solid) from output (dashed).
  (b)~Absorbing $Z$ yields $\Ct$, invariant under all permutations of its legs.}
  \label{fig:sym}
\end{figure}

\noindent Contracting with $\Ct$ alone gives the conjugated product $(xy)^*$;
multiplication is $\Ct$ followed by $Z$, i.e.\ $M = \Ct Z$.
A matrix product $C = AB$ then becomes the real contraction
\begin{equation}\label{eq:realified-product}
  [T_C]_{ij,c}
  = \sum_{k,\,a,\,b,\,c'} [T_A]_{ik,a}\,[T_B]_{kj,b}\;\Ct_{abc'}\,Z_{c'c}
  = \sum_{k,\,a,\,b} [T_A]_{ik,a}\,[T_B]_{kj,b}\;M_{abc}
\end{equation}
(Fig.~\ref{fig:product}), and the same pattern realifies any
pairwise tensor contraction.

\begin{figure}[htb]
  \centering
  \includegraphics[width=0.37\linewidth]{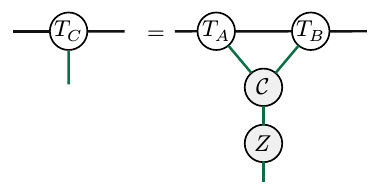}
  \caption{Matrix product $C = AB$ in the realified picture: the green legs
  of $T_A$ and $T_B$ merge through $\Ct$, followed by $Z$ on the output.}
  \label{fig:product}
\end{figure}

\paragraph{Conjugation and phases.}
Quantum circuit networks also require conjugation (bra vectors, adjoints)
and global phases. Both are local operations on the green leg, and neither adds
contraction-stage multiplications: conjugation is a sign flip, and a global phase,
known when the network is built, is absorbed into leaf data during construction.
(A phase applied at run time, e.g.\ a trainable variational parameter, would cost
its own arithmetic.)
Conjugation is the sign flip $Z$: $T_{A^*} = Z T_A$
(Fig.~\ref{fig:conj}(a)). A global phase $e^{i\varphi}$ is the rotation
$R_\varphi = \left(\begin{smallmatrix} \cos\varphi & -\sin\varphi\\
\sin\varphi & \cos\varphi \end{smallmatrix}\right)$
on the same leg, $T_{e^{i\varphi} A} = R_\varphi T_A$
(Fig.~\ref{fig:conj}(b)); pushing phases through the
network is a local gauge transformation.

\begin{figure}[htb]
  \centering
  \includegraphics[width=0.40\linewidth]{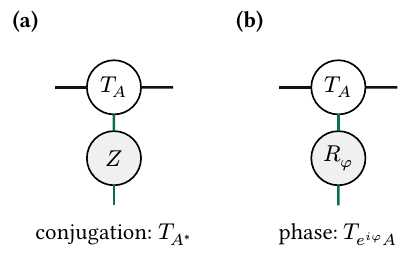}
  \caption{(a)~Conjugation: $Z$ on the green leg gives $T_{A^*}$.
  (b)~Global phase: the rotation $R_\varphi$ on the green leg gives $T_{e^{i\varphi}A}$.
  Neither adds contraction-stage multiplications: the sign flip is free, and
  static phases are absorbed into leaf data at construction.}
  \label{fig:conj}
\end{figure}

\subsection{Tree freedom and the algebraic rules}\label{sec:algebra}

Each pairwise contraction is faithful, so induction over any contraction
tree maps every complex tensor network to a real one of equal value, with
real and imaginary parts delivered on the one open green leg; in
particular, the tree optimizer remains free to choose the cheapest order.
The construction is, however, tree-dependent: realification inserts a
multiplication node $M = \Ct Z$ wherever two green legs merge, and which
legs merge depends on the contraction order, so different trees yield
structurally different real networks from the same complex one
(Fig.~\ref{fig:treefree}). The rest of this subsection recapitulates the local
algebra of $\Ct$~\cite{Bauer2020,BauerNietner2022} that makes these wirings interchangeable, a gauge
choice, and draws two further payoffs: conjugation signs push to the leaves
for free, a property used in the proof of Theorem~\ref{thm:law}, and fixing the
wiring \emph{before} the tree, as in circuit-level realification, is
exposed as exponentially costly (Sec.~\ref{sec:law}).

\begin{figure}[htb]
  \centering
  \includegraphics[width=0.71\linewidth]{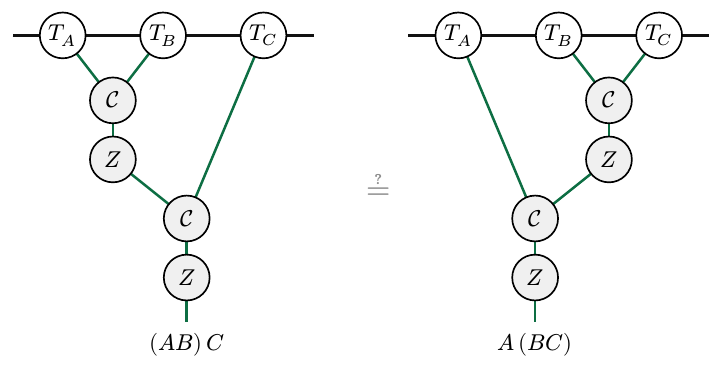}
  \caption{One complex network, two realifications. Contracting the same
  three-tensor chain along $(AB)C$ (left) or $A(BC)$ (right) places the
  $\Ct$ and $Z$ nodes differently, so the two real networks have
  different shapes. Both contract to the same tensor: the wiring is a gauge
  choice, and spider fusion re-aligns it from one tree to the other by
  local moves.}
  \label{fig:treefree}
\end{figure}

\paragraph{Spider fusion.}
The construction rests on \emph{spider fusion}: any tree of $\Ct$ nodes joined by
$Z$-dressed edges, called a \emph{spider}, with one $Z$ on each internal edge, as drawn in
Fig.~\ref{fig:algebra}(d), contracts to the same fully symmetric tensor,
which depends only on the external
legs~\cite{CoeckeDuncan2011,CoeckeKissinger2017} (proof by induction in
Appendix~\ref{app:axioms}). Wirings with cycles are excluded; each closed
loop contributes a scalar factor of $2$, but wirings aligned to contraction
trees are always acyclic. A wiring laid down for one tree can therefore be
transformed to any other by local moves, without changing the result. The
$Z$ factors on the edges cost nothing: they are sign flips, and conjugate
covariance (Eq.~\eqref{eq:rules}) pushes them through $\Ct$ onto the leaves,
where they conjugate the leaf tensors.

Spider fusion rests on four algebraic rules of $\Ct$.
Writing $\mathbf{1} = (1, 0)^\top$ for the real-algebra unit:
\begin{equation}\label{eq:rules}
  \Ct_{abc} = \Ct_{\sigma(a)\sigma(b)\sigma(c)} \;\;\forall \sigma \in S_3, \qquad
  \sum_{a'b'c'} Z_{aa'} Z_{bb'} Z_{cc'}\, \Ct_{a'b'c'} = \Ct_{abc}, \qquad
  \sum_b \Ct_{abc}\, \mathbf{1}_b = Z_{ac},
\end{equation}
together with the cascade rule (Fig.~\ref{fig:algebra}(d)): $(xy)z = x(yz)$.
These are local diagrammatic identities
(Fig.~\ref{fig:algebra}(b)--(e); verification in Appendix~\ref{app:axioms})
that make $(\R^2, \Ct, \mathbf{1})$ a commutative Frobenius algebra of
$\C$ over $\R$~\cite{CoeckeKissinger2017,Bauer2020}.
The Frobenius pairing is
$\beta(x,y) = \operatorname{Re}(xy)$, with Gram matrix $Z$; it links
the multiplication map $M = \Ct Z$ to the fully lowered $\Ct$.
Associativity
is a property of the $Z$-dressed chain used throughout, not of the undressed
$\Ct$ chain (Appendix~\ref{app:axioms}).

\begin{figure}[htb]
  \centering
  \includegraphics[width=0.98\linewidth]{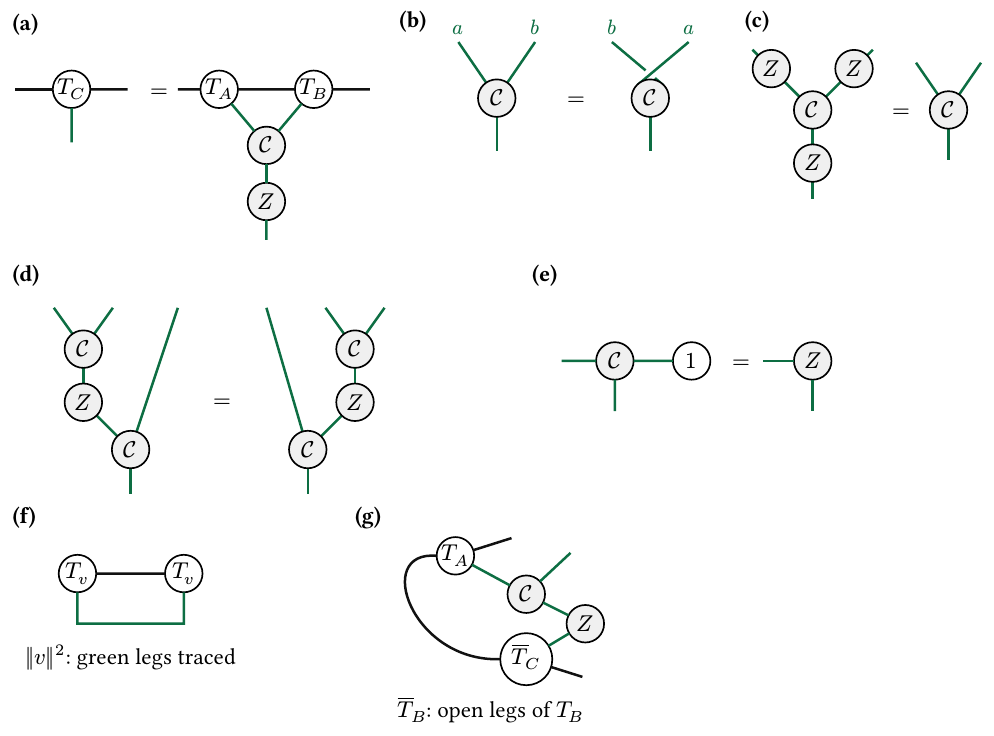}
  \caption{The algebra of the realified representation.
  (a)~Matrix product (see also Fig.~\ref{fig:product}).
  (b)~Permutation invariance: $\Ct$ is symmetric in all three legs.
  (c)~Conjugate covariance: $Z$ on all three legs cancels.
  (d)~Cascade rule (associativity).
  (e)~Unit rule: $\Ct$ contracted with $\mathbf{1} = (1,0)^\top$ reduces to $Z$.
  (f)~Norm square: no $\Ct$ needed.
  (g)~Reverse-mode rule: the pullback $\ov{T}_B$ reproduces
  $\ov{B} = A^\dagger \ov{C}$.}
  \label{fig:algebra}
\end{figure}

\subsection{Cost law and differentiation}\label{sec:law}

We now price the construction and show that reverse-mode differentiation
requires no complex-number extensions.

\begin{theorem}[Cost law]\label{thm:law}
Fix a binary contraction tree for a complex tensor network. Each pairwise contraction is a
\emph{merge} if both operands are green, a \emph{ride} if one is, or a \emph{pass} if neither is;
its \emph{volume} is its scalar-multiplication count, the product of the dimensions of
all indices in the step. Let $m$ and $r$ be the fractions of total volume, measured on
the \emph{real skeleton}, in merges and rides.
The realified contraction on the same tree has arithmetic overhead
\begin{equation}\label{eq:law}
  \textup{overhead} = 1 + 2m + r, \qquad 1 \le \textup{overhead} \le 3,
\end{equation}
counted in real scalar multiplications, and every intermediate's element
count is at most $2\times$ that of its real-skeleton counterpart.
\end{theorem}

\begin{figure}[ht!]
  \centering
  \includegraphics[width=0.35\linewidth]{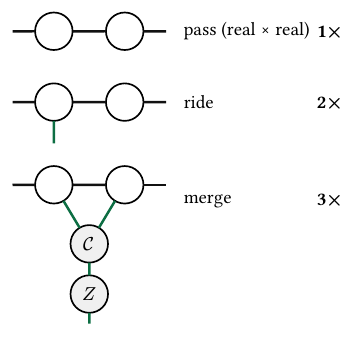}
  \caption{The three pairwise-contraction cases and their arithmetic cost
  relative to the same step on the real skeleton: no green leg ($1\times$, pass),
  a green leg riding on one operand ($2\times$, ride), and a green--green merge
  through $\Ct$ ($3\times$, with the rank-3 kernel of Eq.~\eqref{eq:rank3}).
  }
  \label{fig:cases}
\end{figure}

\begin{proof}
\textit{Memory.}\; Wire the spiders along the tree: where both children of a node carry
green legs, the legs merge through one $\Ct$; where only one does, its leg passes to the
parent. This wiring is tree-aligned, hence contracts to the correct value
(Sec.~\ref{sec:algebra}), and every intermediate carries at most one green leg, one
extra dimension-$2$ index, so no intermediate more than doubles.

\textit{Arithmetic.}\; A pass costs $1\times$. A ride costs $2\times$: the dimension-$2$
green leg is a spectator. A merge evaluates a conjugated product through $\Ct$
(Fig.~\ref{fig:cases}), and its cost is the real tensor rank of $\Ct$ (the minimal
number of real rank-$1$ terms): three real multiplications per complex product, attained
by Gauss's algorithm \cite{Higham1992} and provably minimal \cite{Winograd1971}. Written
as a tensor decomposition of~$\Ct$, with $e_1=(1,0)^\top,\; e_2=(0,1)^\top$ the standard basis of~$\R^2$,
\begin{equation}\label{eq:rank3}
  \Ct = \sum_{k=1}^{3} u_k \otimes v_k \otimes w_k, \qquad
  \begin{gathered}
  (u_1, v_1, w_1) = (e_1,\; e_1,\; e_1{+}e_2), \\
  (u_2, v_2, w_2) = (e_2,\; e_2,\; {-}e_1{+}e_2), \\
  (u_3, v_3, w_3) = (e_1{+}e_2,\; e_1{+}e_2,\; {-}e_2),
  \end{gathered}
\end{equation}
Substituting into Eq.~\eqref{eq:ctensor} recovers both slices. The three terms execute
a merge as three real contractions of the same shape as the real-skeleton step, so a
merge costs $3\times$; contracting the four nonzero entries of $\Ct$ densely would cost
$4\times$. The overhead is the volume-weighted average,
$3m + 2r + (1-m-r) = 1 + 2m + r$, and $m + r \le 1$ gives the bounds.
\end{proof}

Per-GEMM lowering, by contrast, pays a flat $3{\times}$ on every step
(GEMM-3M, Eq.~\eqref{eq:gauss3m}). Network-level realification can only do better: it pays
$3{\times}$ only on merges, while passes and rides are cheaper, so the
overhead $1{+}2m{+}r$ exploits partial realness that flat lowering cannot.

Three remarks sharpen the law.

\paragraph{Remark 1: merge count is a tree invariant.}
Any tree over
$n_c$ complex leaves contains exactly $n_c - 1$ merges, each fusing two green subtrees
into one; optimization controls only \emph{where} the merges land, hence how much volume
they carry (Proposition~\ref{prop:counts}, Appendix~\ref{app:landscape}).

\paragraph{Remark 2: the wiring must follow the tree.}
Fixing the wiring in advance in gate
order, the ancilla worldline of circuit-level realification \cite{Rudolph2002}, and
optimizing the tree around it forces each intermediate to carry one open green leg for
every frozen edge its subtree cuts; each open leg doubles the intermediate, so the
inflation is exponential in the number of cut edges.
For example, a subtree
spanning two disjoint windows of gate order cuts the worldline at four
points, so its intermediate carries four open dimension-2 legs, $2^{4}$
times the elements of its real-skeleton counterpart; restricting every
subtree to one contiguous window cuts only two, but such trees are the
transfer-matrix orders whose intermediates grow with circuit width.
Nor can
the worldline be compressed
away: the ancilla is generically entangled with the data register, its reduced state
mixed unless the amplitudes are real up to a global phase, so each open leg is a genuine
bond.

\paragraph{Remark 3: the law is exact, not a lower bound.}
Eq.~\eqref{eq:law} prices the generic nodewise construction
exactly; it is not a lower bound for every network. Operand structure can undercut
the $3\times$ merge price: when the two green operands are conjugates of each other
and the output is known to be real, as in the norm square of
Fig.~\ref{fig:algebra}(f), the two lanes contract directly at $2\times$ with no
$\Ct$ at all. The count assumes independent operands, no known-real outputs, and no
sharing of products across nodes; such structure only lowers the cost further.

\paragraph{Differentiation.}
The real-valued objectives of circuit simulation, expectation values, norms,
and losses built from inner products and adjoints, are not holomorphic: they depend
on both $z$ and $\bar{z}$. (A bare amplitude is polynomial, hence holomorphic, in
the gate-tensor entries; non-holomorphy enters with the conjugations of a
real-valued objective.) Differentiation in the complex picture therefore requires Wirtinger
calculus~\cite{KreutzDelgado2009}, which tracks $\partial f/\partial z$ and
$\partial f/\partial\bar{z}$ separately. After realification, however, every
tensor is real and every contraction is an ordinary real multilinear map, so
the standard real chain rule applies directly.

Concretely, for a pairwise contraction (Eq.~\eqref{eq:realified-product}),
the reverse-mode pullback with respect to $B$ is
\begin{equation}\label{eq:pullback}
  T_{\bar{B}} = T_A \cdot T_{\bar{C}} \cdot \Ct\, Z,
\end{equation}
the same structure-tensor wiring with the upstream adjoint $T_{\bar{C}}$
replacing the output (Fig.~\ref{fig:algebra}(g)). This reproduces
$\bar{B} = A^\dagger\,\bar{C}$ in the complex picture
(derivation in Appendix~\ref{app:pullback}). Because each node's
pullback is itself a real contraction of the same form, any real-valued
automatic differentiation framework can differentiate the entire realified
network without complex-number extensions.
This coverage is exact for multilinear contraction and its reverse-mode
derivative; spectral factorizations are not included. An ordinary real
singular-value decomposition (SVD) or QR decomposition of
a realified matrix, with the lane index folded into rows or columns, does not
reproduce the complex decomposition (rank and spectrum change: a rank-one complex
matrix can realify to a stacked matrix of real rank two with a different spectrum),
so truncation-based differentiable algorithms~\cite{Liao2019} require a
structure-preserving realification of the factorization itself, which we leave
open. Where decompositions enter, the complex-SVD
adjoint~\cite{Wan2019,Hubig2019} is still needed.

\section{Benchmarks and hardware validation}\label{sec:bench}

We validate the cost law and test the realified representation on
Ascend 910 hardware across 67 circuits: twelve random circuits from the qflex
collection~\cite{Villalonga2019}, and 55 structured instances (Clifford+$T$,
QAOA, and VQE) that span the real-to-complex range.

\paragraph{Setup.} The random circuits are single-amplitude networks
$\langle 0|U|0\rangle$ on rectangular lattices, Bristlecone,
Sycamore~\cite{Arute2019}, and IBM Rochester, spanning 5--70~qubits and
depths~8--32 (Table~\ref{tab:bench}). Three structured families probe the
partially real regime that random circuits do not reach.
\begin{itemize}
\item \emph{Clifford+$T$}: brickwork skeletons of real gates
(H, X, CZ) on a 48-qubit ring and a $6{\times}6$ grid at depth~16, with a
fraction $f_T \in \{0, 0.10, 0.25, 0.50, 1\}$ of one-qubit slots replaced by
$T$ gates; at fixed $f_T$, uniform, spatially clustered, and temporally
clustered placements decouple $T$-count from $T$-position.
\item \emph{QAOA}: MaxCut circuits (64-qubit ring, $6{\times}6$ and $8{\times}8$
grids, 3-regular graphs at $n = 48$) at depths $p = 1$--$8$, contracted as
single-edge $\langle Z_iZ_j\rangle$ expectation networks, with generic and
special ($\pi/4$-multiple) angles.
\item \emph{VQE}: hardware-efficient ans\"atze, class~A ($R_y$ + CZ, all real),
B ($R_y R_z$ + CZ, mixed), C ($R_z R_y R_z$ + CZ, generic complex), on
Ising and Heisenberg chains of 32--64 qubits and a $4{\times}4$ grid, at
depths $L = 2$--$8$, contracted as single Pauli-string expectation networks.
\end{itemize}
We apply no gauge preprocessing (Sec.~\ref{sec:conclusions}) to any circuit.

\paragraph{Orders and statistics.}
Contraction orders come from simulated annealing over binary
trees~\cite{Kalachev2021}, implemented in omeco~\cite{omeco2026}: a
\emph{green-blind} pass (real-skeleton cost as objective) establishes the
baseline, and a \emph{green-aware} pass (Eq.~\eqref{eq:law} as objective)
bounds what reoptimization can recover. An archived \emph{plan} records one
contraction tree and its slicing assignment; pipeline details
are in Appendix~\ref{app:pipeline}. For the structured families, two
independent jobs ran per cell; individual timings are noisy at this
replication depth, so we rely on the sign test (52 of 55 wins; 47 of 50
after collapsing the five QAOA generic/special sibling pairs, one-sided
$p \approx 2 \times 10^{-11}$) rather than per-cell
margins.
The 55 device-clean Ascend cells form the structured benchmark set.

\begin{figure}[tb]
  \centering
  \includegraphics[width=0.8\linewidth]{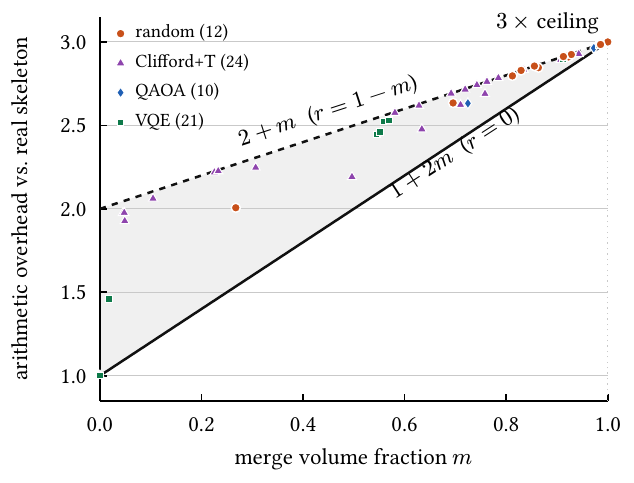}
  \caption{Measured overhead versus merge fraction $m$ for all 67 benchmark instances:
  twelve random circuits (orange circles), 24 Clifford+T (purple triangles),
  10 QAOA (blue diamonds), and 21 VQE (green squares).
  The law \eqref{eq:law} confines every point to the band between the solid line $1 + 2m$
  ($r = 0$) and the dashed line $2 + m$ ($r = 1 - m$).
  All-real circuits (Clifford+T at density~0, TFIM-type VQE) cluster at $(m{=}0,\,1\times)$;
  deep random and QAOA circuits accumulate against the $3\times$ ceiling.}
  \label{fig:law}
\end{figure}

\subsection{The cost law across the circuit suite}\label{sec:lawbench}

The cost law prices the full benchmark suite, and the audits locate each
circuit on its band (Fig.~\ref{fig:law},
Table~\ref{tab:bench}): all 67 instances (twelve random circuits and
55 structured: Clifford+T, QAOA, VQE) track Eq.~\eqref{eq:law} across its
range, from all-real circuits at $m = 0$, overhead $1\times$,
to Sycamore and deep QAOA at the $3\times$ ceiling ($m \approx 1$).
The merge fraction~$m$ is the primary coordinate: it reflects network
topology rather than gate count, a decoupling the structured families probe
directly (Table~\ref{tab:struct}).
The figure shows bimodal clustering: random and deep QAOA circuits
congregate near the $3\times$ ceiling ($m > 0.7$), while all-real circuits
(Clifford+$T$ at density~0, class-A VQE) cluster at $(m{=}0,\,1\times)$.
Engineered placements fill the intermediate band: the spatially clustered
Clifford+$T$ cells hug the upper dashed edge ($r$-dominated, $r = 0.88$),
and the isolated VQE Heisenberg $YY$ point at $(m{=}0.018,\,1.46)$
demonstrates operator-induced complexity on an otherwise all-real ansatz.

\begin{figure}[tb]
  \centering
  \includegraphics[width=0.80\linewidth]{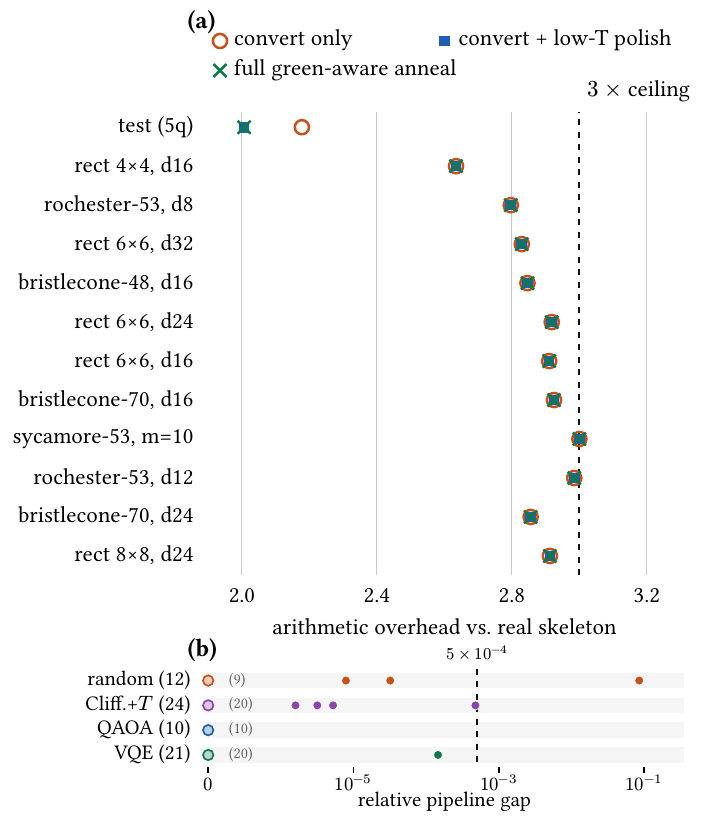}
  \caption{Pipeline comparison: reoptimization after realification is unnecessary
  on the tested suite, within the annealer's search budget.
  (a)~Arithmetic overhead of the twelve random circuits under three strategies:
  reuse the complex-optimal order as-is (open circles),
  refine with a short green-aware polish (filled squares),
  or run a full green-aware anneal (crosses). On eleven circuits the three
  markers overlap to within $10^{-4}$; the sole outlier (5-qubit test) is
  closed by the polish. Dashed line: $3\times$ ceiling.
  (b)~Relative pipeline gap
  $|o_{\text{conv}} - o_{\text{full}}|/o_{\text{full}}$,
  where $o$ denotes the overhead (Eq.~\eqref{eq:law}) under each pipeline,
  for all 67 circuits,
  in rows by family; row labels parenthesize the series size, and the
  left-edge marker is annotated with the count of gaps below $10^{-6}$.
  Open markers at the left edge collect those circuits (most exactly
  zero); filled dots show the remaining gaps. All 55 structured circuits and eleven of
  twelve random circuits fall below $5\times10^{-4}$ (dashed threshold).
  }
  \label{fig:pipe}
\end{figure}

Reoptimization after realification is unnecessary on the tested suite,
within the annealer's search budget.
We compare three pipelines (Fig.~\ref{fig:pipe}): \emph{convert-only},
which realifies the complex-optimal tree without further optimization;
\emph{convert-plus-polish}, which follows with a short low-temperature
green-aware refinement; and a \emph{full green-aware anneal} from scratch
on the realified network. Convert-only agrees with the full anneal to one
part in~$10^4$ on eleven of twelve circuits; the sole exception is the
5-qubit test circuit, where the polish alone closes the gap. Merge volume
is, within the stated search budget, an invariant of the circuit rather
than of the tree: the annealer finds no benefit in trading real-skeleton
quality for better merge placement.

Appendix~\ref{app:landscape} supplies the
theory. A density condition on complex leaves forces flatness
(Theorem~\ref{thm:transfer}), while an explicit green-sparse family shows
that the flatness can fail: the converted tree can require nearly twice as many
multiplications as the realified optimum, though never more than $3\times$
on any network (Theorem~\ref{thm:gap}). The 5-qubit outlier, the green-sparsest of the twelve
random circuits, is that regime in miniature. The same flatness extends to the
structured families: on all 55 device-clean instances, including green-sparse, spatially
clustered ones outside the reach of Theorem~\ref{thm:transfer}, convert-only
stays within $5\times10^{-4}$ of the full green-aware anneal.
All ten QAOA circuits achieve exactly zero gap: their merge-dominated
topology ($m \ge 0.72$, with four of ten at $m \approx 1$) pins the overhead near
the $3\times$ ceiling regardless of tree shape, leaving the annealer
nothing to improve.

The structured families extend the audit across the law's full range
(Table~\ref{tab:struct}). Placement, not gate count, sets the cost: at
$f_T = 0$ the measured overhead is exactly $1.000$; at fixed $f_T = 0.25$,
placement alone moves the ring-48 overhead from $2.93$ (uniform: $m = 0.94$)
through $2.23$ (temporal) to $1.98$ (spatial: $m = 0.05$, $r = 0.88$), as
clustered complex gates convert two-complex-operand merges into
one-complex-operand rides. QAOA saturates the ceiling with depth: ring-64
overhead $2.632$, $2.963$, $2.9999$, $3.000$ at $p = 1, 2, 4, 8$ (the
last two audit-only). Generic and special angles yield identical $(m, r)$ on every
matched pair; for this gate set, special angles change tensor values, not
realness structure. The VQE families interpolate as designed: class~A measures
exactly $1.000$ at every width and depth, classes B and C climb from $2.45$
and $2.52$ at $L = 2$ to $2.99$ at $L = 8$, and a $YY$ expectation on the
all-real class-A circuit isolates operator-induced complexity at $1.460$.

\begin{table}[tb]
  \centering
  \caption{Benchmark summary (random circuits).
  Each row describes one archived winning tree: $\log_2$ real cost is its
  scalar-multiplication count, $(m,r)$ its merge and ride volume fractions.
  The overhead column divides that cost by the cheapest real skeleton found
  across all optimization passes; the last column evaluates
  the law~\eqref{eq:law} from the same tree's own $(m,r)$.
  On eleven of twelve circuits the same tree wins both the skeleton and
  the real-cost searches, so the two columns agree.
  The 5-qubit exception uses a tree whose skeleton is $3.3\%$ costlier
  than the global best, yielding overhead $2.01\times$ against
  that best skeleton but a law value of $1.94\times$ from its own
  fractions; the gap is the skeleton premium.
  Run-to-run variation in the merge/ride split affects the second decimal
  of the last column ($1.93$--$1.97\times$ on reruns); the real cost
  itself is stable.
  The horizontal rule separates the nine core networks from the three
  extension circuits (Appendix~\ref{app:pipeline}).}
  \label{tab:bench}
  \small
  \resizebox{\linewidth}{!}{%
  \begin{tabular}{lrrrrrrr}
    \toprule
    circuit & qubits & complex/total leaves & $\log_2$ real cost & $m$ & $r$ & overhead & $1+2m+r$ \\
    \midrule
    test (5q)            & 5  & 11/40     & 9.25  & 0.268 & 0.408 & $2.01\times$ & $1.94\times$ \\
    rect $4{\times}4$, d16 & 16 & 68/212    & 13.29 & 0.695 & 0.244 & $2.63\times$ & $2.63\times$ \\
    rochester-53, d8     & 53 & 317/767   & 15.93 & 0.812 & 0.172 & $2.80\times$ & $2.80\times$ \\
    rect $6{\times}6$, d16 & 36 & 145/488   & 16.94 & 0.923 & 0.064 & $2.91\times$ & $2.91\times$ \\
    bristlecone-48, d16  & 48 & 209/655   & 16.93 & 0.863 & 0.119 & $2.85\times$ & $2.85\times$ \\
    rect $6{\times}6$, d24 & 36 & 229/660   & 23.34 & 0.918 & 0.082 & $2.92\times$ & $2.92\times$ \\
    bristlecone-70, d16  & 70 & 297/943   & 19.99 & 0.928 & 0.068 & $2.92\times$ & $2.92\times$ \\
    rect $6{\times}6$, d32 & 36 & 302/832   & 29.89 & 0.829 & 0.171 & $2.83\times$ & $2.83\times$ \\
    sycamore-53, 10 cycles & 53 & 1479/1764 & 34.19 & 1.000 & 0.000 & $3.00\times$ & $3.00\times$ \\
    \midrule
    rochester-53, d12    & 53 & 457/1071  & 20.73 & 0.986 & 0.013 & $2.98\times$ & $2.98\times$ \\
    bristlecone-70, d24  & 70 & 448/1273  & 27.24 & 0.856 & 0.144 & $2.86\times$ & $2.86\times$ \\
    rect $8{\times}8$, d24 & 64 & 412/1186 & 30.40 & 0.913 & 0.087 & $2.91\times$ & $2.91\times$ \\
    \bottomrule
  \end{tabular}}%
\end{table}

\paragraph{Complexity budget.}
The realification increment is additive and bounded on the log scale:
$\log_2(1{+}2m{+}r) \le \log_2 3 \approx 1.58$ on every
circuit. Peak space is byte-neutral: the $2\times$ memory law bounds each
intermediate's element count (transient engine workspace lies outside the theorem;
the audited working sets (Table~\ref{tab:audit-ind}) stay consistent with the bound), and since complex64
stores two f32 lanes per element, peak memory in
bytes is identical; on this artifact set no circuit needs slicing.
The comparison that matters on real-only hardware is not the overhead relative
to complex arithmetic but the cost of executing the complex network there
at all: per-GEMM lowering multiplies the complex plan's volume by~$4$ (4M)
or~$3$ (3M), so the realified plan at $1 + 2m + r \le 3$ never adds work
relative to what complex arithmetic already costs on such a device.

As a cross-check, independently re-optimizing each circuit's realified einsum
with the same optimizer policy (Table~\ref{tab:audit-ind},
Appendix~\ref{app:indopt}) produces trees that beat the archived schedule on
all twelve circuits, with realified space within two units of the complex
plan's $\log_2$ peak working set $\mathrm{sc}_{\C}$, and exactly
$\mathrm{sc}_{\C}+1$ on seven of twelve (per-circuit values are archived in
the data repository).

\subsection{Wall-clock comparison}\label{sec:wallclock}

We test whether the arithmetic saving translates to end-to-end speedup
on a real-only accelerator.
All three executors exploit whole-operand realness (the
\emph{structural-real shortcut}): passes and rides skip the merge kernel
entirely. Three matched Ascend 910 executors run each archived plan:
\begin{itemize}
\item \emph{GEMM-4M}: four textbook real products per complex merge
  (multiplication count $1 + 3m + r$).
\item \emph{GEMM-3M}: Gauss's three-multiplication rule at each complex
  merge of the same tree (multiplication count $1 + 2m + r$).
\item \emph{Network-3M}: the rank-3 factorization of
  Eq.~\eqref{eq:rank3} embedded into the tensor network at preparation
  time (multiplication count also $1 + 2m + r$ in the nodewise model of
  Theorem~\ref{thm:law}; the engine's executed factor graph carries a
  bounded loop-volume excess, quantified in
  Appendix~\ref{app:indopt}).
\end{itemize}
\noindent All three share the same f32 inputs, contraction tree, Ascend
\texttt{contract\_binary} primitive, and CANN 8.5 stack; each ran as a
separate single-NPU Slurm job. Reported values are medians of ten samples
after three warmups, admitted only after the precision gates of
Appendix~\ref{app:precision} pass.

\begin{table}[tb]
  \centering
  \caption{Wall-clock medians (ms) on twelve circuits (warmup~3,
  repeats~10, single Ascend~910 NPU, CANN~8.5). All three executors run the same
  archived plans. GEMM-3M applies Gauss's three-multiplication rule
  per merge without rewriting the network; network-3M embeds the
  factorization into the tensor network at preparation time.
  Best time per row in \textbf{bold}; speedup~$= t_{\text{4M}}/t_{\text{net-3M}}$.
  Within-job coefficients of variation
  are below 3\% on every cell except GEMM-3M on rect $4{\times}4$
  ($4.7\%$, one outlier sample).}
  \label{tab:main}
  \small
  \begin{tabular}{lrrrr}
    \toprule
    circuit & GEMM-4M & GEMM-3M & network-3M & speedup \\
    \midrule
    test (5q)                 & 4.50   & 4.91   & \textbf{3.48}  & 1.29 \\
    rect $4{\times}4$, d16    & 38.3   & 38.1   & \textbf{22.0}  & 1.74 \\
    rochester-53, d8          & 170    & 171    & \textbf{97.1}  & 1.75 \\
    rect $6{\times}6$, d16    & 87.0   & 88.3   & \textbf{51.6}  & 1.69 \\
    bristlecone-48, d16       & 117    & 125    & \textbf{72.5}  & 1.61 \\
    rect $6{\times}6$, d24    & 145    & 138    & \textbf{77.9}  & 1.86 \\
    bristlecone-70, d16       & 176    & 173    & \textbf{104}   & 1.68 \\
    rect $6{\times}6$, d32    & 349    & 314    & \textbf{186}   & 1.87 \\
    sycamore-53, 10 cycles    & 883    & 898    & \textbf{785}   & 1.12 \\
    rochester-53, d12         & 244    & 263    & \textbf{139}   & 1.76 \\
    bristlecone-70, d24       & 284    & 303    & \textbf{170}   & 1.67 \\
    rect $8{\times}8$, d24    & 429    & 416    & \textbf{307}   & 1.40 \\
    \midrule
    \multicolumn{4}{l}{median speedup vs.\ GEMM-4M} & 1.68 \\
    \bottomrule
  \end{tabular}
\end{table}

Network-3M was faster than both GEMM baselines on all twelve circuits, with
$t_{\mathrm{4M}}/t_{\mathrm{net\text{-}3M}} = 1.12$--$1.87{\times}$ and a median speedup of
$\approx 1.7{\times}$ (Table~\ref{tab:main}). Within-job
coefficients of variation are below 3\% on all cells but one
(GEMM-3M on rect $4{\times}4$, $4.7\%$), and per-circuit
margins range from 11\% to 47\%.
GEMM-3M, by contrast, runs within $10\%$ of GEMM-4M on every
circuit (median ratio~$1.01$). GEMM-3M and network-3M perform the same
number of multiplications per merge, yet GEMM-3M incurs three separate
kernel launches plus element-wise preparation and combination passes at each node,
whereas network-3M folds the factorization into the contraction graph
so that the execution engine can absorb these costs into fewer, larger
operations. The per-GEMM Gauss rule alone does not
translate into a wall-clock saving; on this software stack, the network-level
rewrite is what realizes it (a hand-fused per-GEMM 3M kernel, which we did not
test, could in principle recover part of the gap).

Because all three executors preserve structural-real shortcuts, the
multiplication-count reference for network-3M versus GEMM-4M is
$(1{+}3m{+}r)/(1{+}2m{+}r)$, with median
$1.31$; the measured speedup exceeds this reference on eleven circuits and
falls below it on sycamore-53, while remaining above parity in every case.
The multiplication count predicts but does not bound wall-clock time:
additions, materialization, and data movement are not separately accounted
for; the excess of the measured speedups over this reference is not
decomposed into launch, materialization, and data-movement contributions in
this work, and per-kernel profiling of the three executors is left to future
work.

The structured families confirm the wall-clock advantage
(Table~\ref{tab:struct}): network-3M beat GEMM-4M on 52 of 55 device-clean
Ascend cells, with median
speedup $t_{\mathrm{4M}}/t_{\mathrm{net\text{-}3M}} = 1.42{\times}$ (range $0.90$--$1.96{\times}$)
and family medians $1.41{\times}$ (Clifford+$T$), $1.80{\times}$ (QAOA), $1.34{\times}$ (VQE).
GEMM-3M again shows no systematic gain over GEMM-4M (family medians
$0.98$--$1.00$; per-cell ratios $0.83$--$1.40$). The three non-wins are mild slowdowns
($t_{\mathrm{4M}}/t_{\mathrm{net\text{-}3M}} = 0.90$--$0.99{\times}$, i.e.\ at most
$12\%$): two $f_T{=}1$ circuits whose predicted
$\sim\!20\%$ multiplication saving is absorbed by
execution-path overhead at their modest contraction volumes, and
one all-real ansatz cell at multiplication parity ($m = r = 0$).

Two controls bound what the structured numbers can claim. The nine device-clean all-real
cells execute identical multiplication counts in every executor at
$m = r = 0$, yet $t_{\mathrm{4M}}/t_{\mathrm{net\text{-}3M}} = 0.99$--$1.46$
(median $1.16$): these cells isolate the execution-path contribution that is
folded into every measured ratio. And on an A800 with native complex
arithmetic, the same plans run fastest in native complex form on 44 of 55
rows (family medians
$t_{\mathrm{4M}}/t_{\mathrm{net\text{-}3M}} = 0.75$--$0.83$): the
wall-clock advantage is specific to the real-only accelerator regime. Peak
memory remained byte-neutral on all 55 device-clean instances.

\begin{table}[tb]
  \centering
  \caption{Structured-circuit excerpt: fifteen of the 55 device-clean cells.
  All randomized instance parameters in the rows shown ($T$-gate
  placements, rotation angles) were drawn with pseudorandom seed~$0$; QAOA rows use generic angles, and
  the special-angle siblings are structurally identical. Wall-clock medians (ms) on the same
  Ascend~910/CANN~8.5 stack as Table~\ref{tab:main}, warmup 3, repeats 10,
  median over two independent jobs. GEMM-3M applies Gauss's
  three-multiplication rule per merge at run time, without rewriting the
  network. The full structured-instance matrix, including GPU and native-complex
  columns, excluded cells, and fallback flags, is in the
  data repository. Best time per row in \textbf{bold};
  speedup~$= t_{\text{4M}}/t_{\text{net-3M}}$.}
  \label{tab:struct}
  \footnotesize
  \setlength{\tabcolsep}{4pt}
  \begin{tabular}{@{}lrrrrrrr@{}}
    \toprule
    circuit & $m$ & $r$ & $1{+}2m{+}r$ & GEMM-4M & GEMM-3M & network-3M & speedup \\
    \midrule
    C+$T$ ring-48, $f_T{=}0$                      & 0.000 & 0.000 & 1.000 & 110.2 & 109.0 & \textbf{93.46} & 1.18 \\
    C+$T$ ring-48, $f_T{=}0.25$ spatial           & 0.048 & 0.881 & 1.977 & 131.8 & 144.0 & \textbf{105.4} & 1.25 \\
    C+$T$ ring-48, $f_T{=}0.25$ temporal          & 0.233 & 0.760 & 2.226 & 161.2 & 163.8 & \textbf{118.1} & 1.36 \\
    C+$T$ ring-48, $f_T{=}0.25$ uniform           & 0.943 & 0.043 & 2.928 & 188.2 & 186.0 & \textbf{117.1} & 1.61 \\
    C+$T$ grid-$6{\times}6$, $f_T{=}0.50$ uniform & 0.985 & 0.015 & 2.985 & 146.3 & 145.5 & \textbf{97.75} & 1.50 \\
    \midrule
    QAOA ring-64, $p{=}1$                         & 0.725 & 0.183 & 2.632 & 114.7 & 112.3 & \textbf{65.53} & 1.75 \\
    QAOA ring-64, $p{=}2$                         & 0.972 & 0.018 & 2.963 & 212.7 & 208.8 & \textbf{120.6} & 1.76 \\
    QAOA grid-$6{\times}6$, $p{=}1$               & 0.989 & 0.007 & 2.985 & 93.39 & 93.44 & \textbf{47.57} & 1.96 \\
    QAOA grid-$6{\times}6$, $p{=}2$               & 1.000 & 0.000 & 3.000 & 232.5 & 209.5 & \textbf{143.0} & 1.63 \\
    QAOA grid-$8{\times}8$, $p{=}1$               & 0.999 & 0.001 & 2.999 & 180.2 & 165.7 & \textbf{94.93} & 1.90 \\
    \midrule
    VQE-A TFIM-32, $L{=}8$ (all real)             & 0.000 & 0.000 & 1.000 & 114.2 & 104.0 & \textbf{86.09} & 1.33 \\
    VQE-A Heis-32, $L{=}4$, $YY$ term             & 0.018 & 0.424 & 1.460 & 54.94 & 50.82 & \textbf{40.67} & 1.35 \\
    VQE-B TFIM-32, $L{=}2$                        & 0.545 & 0.356 & 2.446 & 82.48 & 80.46 & \textbf{52.30} & 1.58 \\
    VQE-B TFIM-32, $L{=}8$                        & 0.985 & 0.015 & 2.985 & 405.9 & 380.4 & \textbf{230.6} & 1.76 \\
    VQE-C TFIM-32, $L{=}4$                        & 0.919 & 0.073 & 2.911 & 210.6 & 223.6 & \textbf{156.9} & 1.34 \\
    \bottomrule
  \end{tabular}
\end{table}

\paragraph{F32 accuracy.}\label{sec:precision}
The Ascend backend is f32-only. Table~\ref{tab:precision} verifies that
the f32 amplitudes are accurate by comparing each component of the
Ascend network-3M result with a ComplexF64 reference (\texttt{CudaComplex<f64>}
on an A800~GPU via the cuQuantum stack~\cite{Bayraktar2023}, same archived plan). Real and imaginary parts each
differ by less than $0.002\%$ on every circuit, as does the normwise error
$|\Delta z|/|z|$ (Appendix~\ref{app:precision}).

\begin{table}[tb]
  \centering
  \caption{Component-wise and normwise relative difference (\%) between the
  Ascend~910 network-3M f32 result and the A800 ComplexF64 reference for the
  single-amplitude contraction $\langle 0|U|0\rangle$.
  Each result is one complex scalar; the columns report
  $|\Delta\mathrm{Re}|/|\mathrm{Re}_{\mathrm{ref}}|$,
  $|\Delta\mathrm{Im}|/|\mathrm{Im}_{\mathrm{ref}}|$, and
  $|\Delta z|/|z_{\mathrm{ref}}|$, all computed from the archived amplitude
  pairs; the normwise error also stays below $0.002\%$ on every circuit.}
  \label{tab:precision}
  \small
  \begin{tabular}{lccc}
    \toprule
    circuit & Re (\%) & Im (\%) & $|\Delta z|/|z|$ (\%) \\
    \midrule
    test (5q)              & $4.2 \times 10^{-5}$ & $2.8 \times 10^{-5}$ & $2.9 \times 10^{-5}$ \\
    rect $4{\times}4$, d16  & $3.5 \times 10^{-4}$ & $3.3 \times 10^{-4}$ & $3.3 \times 10^{-4}$ \\
    rochester-53, d8       & $1.2 \times 10^{-3}$ & $1.2 \times 10^{-3}$ & $1.2 \times 10^{-3}$ \\
    rect $6{\times}6$, d16  & $1.2 \times 10^{-3}$ & $7.4 \times 10^{-4}$ & $7.7 \times 10^{-4}$ \\
    bristlecone-48, d16    & $1.0 \times 10^{-3}$ & $1.3 \times 10^{-3}$ & $1.0 \times 10^{-3}$ \\
    rect $6{\times}6$, d24  & $1.1 \times 10^{-3}$ & $1.0 \times 10^{-3}$ & $1.1 \times 10^{-3}$ \\
    bristlecone-70, d16    & $1.2 \times 10^{-3}$ & $1.5 \times 10^{-3}$ & $1.3 \times 10^{-3}$ \\
    rect $6{\times}6$, d32  & $1.1 \times 10^{-3}$ & $9.5 \times 10^{-4}$ & $1.1 \times 10^{-3}$ \\
    sycamore-53, 10 cycles & $8.0 \times 10^{-4}$ & $1.2 \times 10^{-3}$ & $8.3 \times 10^{-4}$ \\
    \midrule
    rochester-53, d12      & $1.9 \times 10^{-3}$ & $1.7 \times 10^{-3}$ & $1.8 \times 10^{-3}$ \\
    bristlecone-70, d24    & $1.7 \times 10^{-3}$ & $1.9 \times 10^{-3}$ & $1.9 \times 10^{-3}$ \\
    rect $8{\times}8$, d24  & $1.96 \times 10^{-3}$ & $1.9 \times 10^{-3}$ & $1.9 \times 10^{-3}$ \\
    \bottomrule
  \end{tabular}
\end{table}

\FloatBarrier

\section{Conclusions}\label{sec:conclusions}
We formulated tensor-network contraction over the complex field as a real
contraction: each complex tensor acquires one auxiliary dimension-2 index.
The tight cost law $1 + 2m + r$, with every intermediate's size at most
$2\times$ that of its real-skeleton counterpart, is proved and audited
across all 67 benchmark circuits.
Convert-only contraction orders match full reoptimization to within $5\times10^{-4}$ on 66 of
67 circuits, the exception closed by a cheap polish.
On Ascend 910, the resulting network-3M executor was faster than both GEMM-4M
and GEMM-3M on all twelve random circuits and on 52 of 55 device-clean
structured cells (the three non-wins at $0.90$--$0.99{\times}$), with median
speedup $t_{\mathrm{4M}}/t_{\mathrm{net\text{-}3M}} = 1.68{\times}$ on random circuits and $1.42{\times}$ on structured families (arithmetic and engine-path components separated in Sec.~\ref{sec:wallclock}).
GEMM-3M, which applies the same three-multiplication rule per merge but
without rewriting the network, tracks GEMM-4M in the medians ($1.01$
random, $0.98$--$1.00$ structured), with per-cell ratios spanning
$0.83$--$1.40$; the network-level rewrite, not the arithmetic
identity alone, drives the speedup.

The advantage over per-GEMM complex lowering is structural, not only
arithmetic: a dispatcher with operand-realness checks matches the
multiplication count of Eq.~\eqref{eq:law}, but only the realified network is
a single real einsum, so reverse-mode differentiation uses the standard
real-valued rules that existing frameworks provide (Sec.~\ref{sec:law};
factorizations such as SVD remain open).

Against the flat $4\times$ of per-step GEMM-4M lowering (the textbook cost
when complex arithmetic is reconstructed from real GEMMs, as on
TPUs~\cite{Hauru2021,Morningstar2022,Ganahl2023}) and the flat $3\times$ of uniform
GEMM-3M~\cite{Higham1992}, realification pays $1+2m+r$: a saving of
$1.33$--$2.06\times$ in total multiplications on the random suite and up to
$4\times$ in the all-real limit. Precision management (bfloat16
splitting~\cite{Lewis2022}, compensated real-GEMM
accumulation~\cite{Ootomo2023}, and the 3M stability analysis of~\cite{Higham1992})
composes with the representation unchanged, since every kernel is a real GEMM.
Beyond simulation, the representation extends to any complex tensor
network, including expectation values of projected entangled pair states
(PEPS) and structured complex linear algebra, and carries its gauge freedom and real-valued differentiation rule
with it. Two extensions remain open: gauge
preprocessing can reduce the number of complex leaves, and
slicing for distributed contraction commutes with the construction, though the
interaction with optimal slicing has not been formally characterized.

Several caveats qualify the evidence: random-circuit cells were timed in
one Slurm job each and structured cells in two, so per-cell margins
carry run-to-run noise; the aggregate win records (all twelve random
circuits; 52 of 55 structured cells, 47 of 50 after deduplicating the
QAOA sibling pairs, sign test $p \approx 2 \times 10^{-11}$)
are the primary robustness evidence; QAOA and VQE cells contract one representative Pauli-string
expectation each, not full Hamiltonians or optimization loops; reverse-mode
differentiation of the realified network has not been benchmarked; all
comparisons are at f32 precision; and the wall-clock advantage is specific to
real-only accelerators, the same plans running fastest in native complex form
on an A800 (44 of 55 device-clean rows at these problem sizes).

\paragraph{Data availability.} The archived contraction plans, experiment manifests, wall-clock measurements, correctness-gate outputs, and analysis scripts supporting the findings of this article are publicly available at \url{https://doi.org/10.5281/zenodo.21791682}, with the full Rust toolchain (yao-rs, omeco, and omeinsum-rs) pinned by commit SHA.

\paragraph{Funding information.} This work was partially supported by the National Key R\&D
Program of China (Grant No.~2024YFB4504004). This work was also supported by the National Natural
Science Foundation of China (Grant No.~12404568).

\appendix

\section{The reoptimization landscape}\label{app:landscape}

Section~\ref{sec:lawbench} reports that contraction orders optimized on the
complex network survive realification essentially unchanged on eleven of
twelve random circuits. This appendix explains that observation: it identifies
which quantities are invariants of the tree, states when flatness of the
landscape is forced, and shows by explicit construction that the flatness is
not universal.

Throughout, fix a network with $n$ leaves, $n_c \ge 1$ of them green (complex)
and $n_r = n - n_c$ real. For a binary contraction tree $T$, let $V(T)$ be
its real-skeleton volume, $m(T), r(T)$ its merge and ride volume fractions,
and $C(T) = \bigl(1 + 2 m(T) + r(T)\bigr) V(T)$ its realified cost
(Theorem~\ref{thm:law}). Let $V^* = \min_T V(T)$ and
$C^* = \min_T C(T)$.

\begin{proposition}[Invariance of step counts]\label{prop:counts}
Every binary contraction tree over the network has exactly $n_c - 1$ merge
steps and exactly $n_r$ non-merge steps (rides and passes combined). Only the
split of the $n_r$ non-merge steps into rides and passes, and the volumes
carried by all steps, depend on the tree.
\end{proposition}

\begin{proof}
Induction on the tree. A leaf has no steps. For an internal root whose
children subtrees contain $n_c^{(1)}$ and $n_c^{(2)}$ green leaves: if both
are positive, the root is a merge and the totals are
$(n_c^{(1)} - 1) + (n_c^{(2)} - 1) + 1 = n_c - 1$ merges; if one is zero, the
root is not a merge and the merge count is $n_c - 1$ from the green side
alone (a subtree with no green leaves contains no merges). The non-merge count is then $(n - 1) - (n_c - 1) = n_r$. That the
ride/pass split varies is shown by the three-leaf network $g, x_1, x_2$ with
$g$ green: the tree $(g x_1) x_2$ has two rides, while $g (x_1 x_2)$ has one
pass and one ride.
\end{proof}

\noindent Optimization therefore never changes \emph{how many} merges occur,
only \emph{where} they land, i.e.\ which volumes they carry. Flatness of the
landscape is a statement about volumes, and it holds exactly when no cheap
tree can shelter large volumes from the merges.

\begin{theorem}[Density implies flatness]\label{thm:transfer}
Let $0 \le \epsilon \le 1/2$, and suppose the network satisfies the density
hypothesis
\begin{equation}\label{eq:density}
  m(T) \;\ge\; 1 - \epsilon
  \qquad \text{for every tree } T \text{ with } V(T) \le 3 V^*.
\end{equation}
Then every skeleton-optimal tree $T_0$ satisfies
\begin{equation}
  \frac{C(T_0)}{C^*} \;\le\; \frac{3}{3 - 2\epsilon} \;\le\; 1 + \epsilon:
\end{equation}
green-blind optimization followed by conversion is $(1+\epsilon)$-optimal for
the realified objective, and reoptimization cannot recover more than a factor
$3/(3-2\epsilon)$.
\end{theorem}

\begin{proof}
Since overheads lie in $[1, 3]$, $C(T_0) \le 3 V(T_0) = 3 V^*$. For any tree
$T$: if $V(T) > 3V^*$ then $C(T) \ge V(T) > 3V^*$; otherwise
Eq.~\eqref{eq:density} applies, and with $p + r + m = 1$ (pass fraction $p$),
$1 + 2m + r = 3 - 2p - r \ge 3 - 2(1-m) \ge 3 - 2\epsilon$, so
$C(T) \ge (3 - 2\epsilon) V(T) \ge (3 - 2\epsilon) V^*$. Hence
$C^* \ge (3-2\epsilon) V^*$ and the ratio follows;
$3/(3-2\epsilon) \le 1+\epsilon$ for $\epsilon \le 1/2$.
\end{proof}

\begin{corollary}[Checkable sufficient condition]\label{cor:crude}
For any tree $T$, $1 - m(T) - r(T) \le 1 - m(T) \le n_r\, v_{\max}(T) / V(T)$,
where $v_{\max}(T)$ is the largest single-step skeleton volume: by
Proposition~\ref{prop:counts} the non-merge volume is spread over exactly
$n_r$ steps. Thus Eq.~\eqref{eq:density} holds with
$\epsilon = n_r \sup_T v_{\max}(T)/V(T)$, the supremum over trees with
$V(T) \le 3V^*$.
\end{corollary}

\noindent The corollary is crude (on deep circuits the volume concentrates
in few steps and the bound can exceed $1$) but it isolates the mechanism:
non-merge volume is confined to as many steps as there are \emph{real}
leaves. When green leaves are dense, every large step has green leaves on
both sides in every cheap tree, and the landscape is pinned to the
$3\times$ ceiling. On the benchmark suite the hypothesis is verified within
search budget rather than proved: the green-aware anneal of
Sec.~\ref{sec:lawbench} sampled no tree with $V(T) \le 3V^*$ and small $m$,
consistent with merge fractions $0.695$--$1.00$ across the eleven flat
circuits.

Flatness is not universal. The following construction shows that on
green-sparse networks, the converted skeleton optimum can require nearly
twice as many real multiplications as the realified optimum (memory is not
at stake: the $2\times$ bound of Theorem~\ref{thm:law} is tree-independent).
The reoptimization pass that the tested suite renders unnecessary is therefore
essential in the worst case.

\begin{theorem}[Conversion gap]\label{thm:gap}
For every $\epsilon > 0$ there is a network family with $n_c = 2$ green
leaves and a \emph{unique} skeleton-optimal tree $T_0$ such that
\[
  C(T_0) \;\ge\; (2 - \epsilon)\, C^*.
\]
No ratio above $3$ is possible for any network, since $C(T_0) \le 3V^* \le 3\,C^*$.
\end{theorem}

\begin{proof}
Fix $\chi \ge \max(3, \lceil 4/\epsilon \rceil)$ and take $N$ large. Take the open
chain $A_1 A_2 \cdots A_N$ with bonds $a_i$ between $A_i$ and $A_{i+1}$ of
dimension $\chi$, an open leg $a_0$ of dimension $\chi - 1$ on $A_1$, and an
open leg $a_N$ of dimension $\chi$ on $A_N$. Leaves $A_1, A_2$ are green, the
rest real.

\emph{Unique skeleton optimum.} Consider any step of any binary tree: a
contraction of disjoint leaf sets $P$ and $Q$, of cost the product of the
open legs of $P \cup Q$ and the bonds joining $P$ to $Q$. A leaf set with $f$
maximal contiguous fragments has exactly $2f$ open legs (each fragment ends
either at a bond to its complement or at $a_0$ or $a_N$), and $a_0$ is the
only leg of dimension $\chi - 1$. A step of cost $(\chi-1)\chi^2$ must
therefore expose exactly the legs $\{a_0, a_j\}$ and contract one bond:
$P \cup Q$ is a prefix interval containing $A_1$, split by one bond into two
intervals; the step absorbs an adjacent block into the prefix containing
$A_1$. Every other step costs at least $\chi^3$: a one-fragment union without
$A_1$ exposes two $\chi$-legs and contracts at least one bond
($\ge \chi^3$); a fragmented union ($f \ge 2$) exposes at least four legs, at
most one of dimension $\chi - 1$ ($\ge (\chi-1)\chi^3 \ge 2\chi^3$); a prefix
union whose parts interleave contracts at least two bonds
($\ge (\chi-1)\chi^3$). A tree all of whose $N-1$ steps are of the cheap kind
grows a prefix interval from $A_1$ one leaf at a time: it is exactly the
left-to-right caterpillar $T_0 = (\cdots((A_1 A_2) A_3) \cdots) A_N$, with
$V(T_0) = (N-1)(\chi-1)\chi^2$. Any other binary tree has at least one step
costing $\chi^3$ or more, hence
$V \ge (N-2)(\chi-1)\chi^2 + \chi^3 = V(T_0) + \chi^2 > V(T_0)$. So $T_0$ is
the unique skeleton optimum and $V^* = V(T_0)$.

\emph{Cost of the converted tree.} In $T_0$ the first step $A_1 A_2$ is a
merge, and every later step rides the green block:
$C(T_0) = 3(\chi-1)\chi^2 + 2(N-2)(\chi-1)\chi^2 = (2N-1)(\chi-1)\chi^2$.

\emph{A quarantining tree.} Let $T_q$ contract $A_3 \cdots A_N$
left-to-right first ($N-3$ passes of cost $\chi^3$), then absorb $A_2$ (one
ride, $2\chi^3$), then $A_1$ (one merge, $3(\chi-1)\chi^2$):
$C(T_q) \le (N-1)\chi^3 + 3\chi^3$. Hence
\[
  \frac{C(T_0)}{C^*} \;\ge\; \frac{C(T_0)}{C(T_q)}
  \;\ge\; \frac{(2N-1)(\chi-1)}{(N+2)\,\chi}
  \;\xrightarrow{N \to \infty}\; 2 - \frac{2}{\chi} \;\ge\; 2 - \frac{\epsilon}{2},
\]
so every sufficiently large $N$ attains the claimed $2 - \epsilon$.
\end{proof}

\noindent The mechanism is the ride channel: skeleton optimality forces the
green head of the chain to ride every step (overhead $\to 2$), while
quarantining the two green leaves until the end costs only a
$\chi/(\chi-1)$ skeleton premium that vanishes for large $\chi$. The
quarantining tree is contiguous (no outer products are needed), so the gap
is visible to the optimizer of Sec.~\ref{sec:bench}. Whether a family can
force the \emph{merge} channel onto the dominant volume of every
near-optimal skeleton tree while remaining avoidable at vanishing premium,
pushing the ratio toward the trivial ceiling $3$, is open; our attempts at
such constructions ran into a trade-off (embedding green tensors deep enough
to force merges everywhere also forces every competing tree to pay the ride
factor on the same volume, capping the gap at $3/2$), and we conjecture the
true supremum lies strictly below $3$.

\paragraph{The search space contains the gap trees.}
Both annealing passes draw moves from the rotation family: at an internal
edge with subtrees $x$ (sibling), $b, c$ (children), the move exchanges $x$
with one child, a nearest-neighbor interchange (NNI). NNI moves connect the space
of binary trees over a fixed leaf set~\cite{Robinson1971}, and neither the
move generator nor the cost evaluator excludes steps between operands with no
shared index: outer products are proposed and priced like any other step
(their volume is the product of all open dimensions). The flatness observed
in Sec.~\ref{sec:lawbench} is therefore not an artifact of an excluded move
class. The standard caveat applies: connectivity does not imply rapid mixing,
so annealing certifies only what it samples.

\paragraph{Where the benchmarks sit.}
The two theorems bracket the observations of Sec.~\ref{sec:lawbench}. Eleven
circuits carry merge fractions $0.695$--$1.00$ and behave as
Theorem~\ref{thm:transfer} predicts: convert-only
matches the full anneal to $10^{-4}$. The 5-qubit test circuit is the
green-sparsest of the twelve random circuits ($n_c/n = 11/40$, $m = 0.27$) and shows
the gap mechanism in miniature: the winning realified tree buys better merge
placement at a $+3.3\%$ skeleton premium (Table~\ref{tab:bench}). The channel
differs (Theorem~\ref{thm:gap}'s worst case is carried by rides, the 5-qubit
gain by merge placement) but the trade is the same: a small skeleton premium
buys a cheaper green-touched volume.
The structured families widen the sample: all 55 device-clean Clifford+$T$, QAOA, and VQE
instances, including the spatially clustered, green-sparse ones, stay within
$5\times10^{-4}$ of the full anneal, so the gap regime of
Theorem~\ref{thm:gap} remains unobserved outside adversarial constructions.

\section{Independent optimization cross-check}\label{app:indopt}

Table~\ref{tab:audit-ind} reports realified complexities from independently
optimized contraction orders (TreeSA, same policy as the main campaign). Here
$\mathrm{tc}_{\R}$, $\mathrm{sc}_{\R}$, and $\mathrm{rwc}_{\R}$ are the
optimizer's time complexity (total loop iterations, in $\log_2$), space
complexity (peak working set, in $\log_2$ elements), and read--write cost
(in $\log_2$ of the byte volume touched), all in the optimizer's loop-volume
convention; $\mathrm{rwc}_{\C}$ is the read--write cost of the corresponding
complex plan. The $\Delta\mathrm{tc}$ excess over the same-tree law value is a
predictable convention effect.
The cost law charges three skeleton-step loop volumes per merge
(one per rank-1 term of Eq.~\eqref{eq:rank3}); the optimizer additionally
counts the contractions of each length-2 factor vector
$u_k,v_k,w_k$ against the operands.
For a merge step whose operands have free-index dimensions $D_i,D_j$
and contraction-index dimension $D_k$ on the skeleton, the ratio of the
factor-graph loop volume to the skeleton step is
\begin{equation}\label{eq:kappa}
  \kappa = 3 + 6\!\left(\frac{1}{D_i}+\frac{1}{D_j}+\frac{1}{D_k}\right).
\end{equation}
Pass and ride steps carry no excess ($\kappa=1$ and $\kappa=2$,
respectively, matching their multiplication counts).
The $6/D$ correction is the cost of the dim-2
projections---$O(D^2)$ per step versus $O(D^3)$ for the
product---so $\kappa\to 3$ for large tensors and the two conventions
converge; the aggregate $\Delta\mathrm{tc}$ is a volume-weighted average
of the per-step excesses.
The table confirms the trend: the 5-qubit test circuit
(small tensors, large $\kappa$) shows 95\% excess, while the largest
circuits fall to 35\% or below.
All twelve independently optimized trees beat the archived executed schedule.

\begin{table}[tb]
  \centering
  \caption{Independently optimized realified complexities in the optimizer's
  loop-volume convention: $\mathrm{tc}_{\R}$ is the time complexity (total
  loop iterations), $\mathrm{sc}_{\R}$ the space complexity (peak working
  set), and $\mathrm{rwc}_{\R}$ the read--write cost (byte volume touched),
  all in $\log_2$; $\mathrm{rwc}_{\C}$ is the read--write cost of the
  corresponding complex plan. $\Delta\mathrm{tc}$ is the excess over the
  same-tree law value of Table~\ref{tab:bench}, reported as $2^{\Delta\mathrm{tc}}-1$
  in percent in the penultimate column; it is a predictable convention effect
  (Eq.~\ref{eq:kappa}: multiplication count versus loop volume), not an
  optimization deficit.
  Percentages are computed from unrounded $\Delta\mathrm{tc}$, so equal
  displayed $\Delta\mathrm{tc}$ values can map to slightly different
  percentages.}
  \label{tab:audit-ind}
  \small
  \begin{tabular}{lrrrrrr}
    \toprule
    circuit & $\mathrm{tc}_{\R}$ & $\mathrm{sc}_{\R}$ & $\mathrm{rwc}_{\R}$ & $\Delta\mathrm{tc}$ & $\Delta\mathrm{tc}$ excess (\%) & $\mathrm{rwc}_{\R}{-}\mathrm{rwc}_{\C}$ \\
    \midrule
    test (5q)            & 10.21 & 5.0  & 10.56 & $+0.96$ & $+95\%$ & $+1.51$ \\
    rect $4{\times}4$, d16 & 14.11 & 9.0  & 13.99 & $+0.82$ & $+77\%$ & $+1.76$ \\
    rochester-53, d8     & 16.76 & 10.0 & 16.34 & $+0.83$ & $+78\%$ & $+1.87$ \\
    rect $6{\times}6$, d16 & 17.31 & 11.0 & 16.35 & $+0.37$ & $+30\%$ & $+1.49$ \\
    bristlecone-48, d16  & 17.45 & 11.6 & 16.66 & $+0.51$ & $+43\%$ & $+1.45$ \\
    rect $6{\times}6$, d24 & 23.63 & 18.0 & 21.97 & $+0.29$ & $+22\%$ & $+1.32$ \\
    bristlecone-70, d16  & 20.36 & 14.6 & 18.70 & $+0.37$ & $+29\%$ & $+1.43$ \\
    rect $6{\times}6$, d32 & 30.24 & 24.0 & 28.28 & $+0.35$ & $+28\%$ & $+0.86$ \\
    sycamore-53, m=10    & 34.61 & 27.0 & 31.51 & $+0.42$ & $+33\%$ & $+0.88$ \\
    \midrule
    rochester-53, d12    & 21.18 & 14.0 & 19.15 & $+0.45$ & $+37\%$ & $+1.10$ \\
    bristlecone-70, d24  & 27.52 & 22.0 & 25.82 & $+0.29$ & $+22\%$ & $+0.90$ \\
    rect $8{\times}8$, d24 & 30.83 & 24.0 & 28.55 & $+0.43$ & $+35\%$ & $+0.66$ \\
    \bottomrule
  \end{tabular}
\end{table}

\section{Software pipeline and reproducibility}\label{app:pipeline}

Circuits are single-amplitude networks $\langle 0|U|0\rangle$, each contracting
to one scalar amplitude of the circuit unitary~$U$, built with the Yao
framework~\cite{Luo2020} (see also TensorCircuit~\cite{Zhang2023}). Circuit
labels give the layout and depth: d8--d32 counts layers, and Sycamore is labeled
by cycles~$m$. Eight qflex networks~\cite{Villalonga2019} plus a 5-qubit test circuit
vendored with the suite form the core nine; three extension circuits bring
the total to twelve. Cost is the einsum multiplication count, the
scalar-multiplication volume at
uniform bond dimension~2. Leaves are classified \emph{green} when their gate
tensor has a nonzero imaginary part (Sec.~\ref{sec:rep}); the test is exact and is applied once at network
construction: any bitwise-nonzero f64 entry of the vendored imaginary block
marks the leaf, with no tolerance and no gauge
preprocessing.

Every number in Sec.~\ref{sec:bench} is produced by an all-Rust pipeline:
circuits are parsed and converted to tensor networks by yao-rs, contraction
orders are optimized by omeco, a simulated-annealing optimizer over binary
contraction trees (tree simulated annealing, TreeSA~\cite{Kalachev2021},
initialization, then the green-aware anneal), and plans are archived, audited, and
executed by omeinsum-rs, all pinned by commit SHA in the NPUBenchmarkData
repository.

\section{Verification of the algebraic rules}\label{app:axioms}

The four rules, the three identities of Eq.~\eqref{eq:rules} together with
the cascade rule, are verified from the closed form
$\Ct_{abc} = \operatorname{Re}(i^{\,a+b+c-3})$ of Eq.~\eqref{eq:ctensor};
spider fusion (Sec.~\ref{sec:algebra}) then follows by induction on the tree.

\emph{Permutation invariance.}\; The closed form depends only on the sum $a+b+c$, so
$\Ct_{abc} = \Ct_{\sigma(a)\sigma(b)\sigma(c)}$ for all $\sigma \in S_3$.

\emph{Conjugate covariance ($Z^{\otimes 3}$).}\; $Z$ maps index $1 \mapsto 1$,
$2 \mapsto -2$, so $\sum Z_{aa'}Z_{bb'}Z_{cc'}\Ct_{a'b'c'}$ picks up a factor
$(-1)^{\#\{a,b,c\,=\,2\}}$; since $\Ct_{abc}$ vanishes unless $a+b+c-3$ is even,
the sign is always $+1$.

\emph{Unit rule.}\; $\sum_b \Ct_{abc}\,\mathbf{1}_b = \Ct_{a1c}$
since $\mathbf{1} = (1,0)^\top$. Direct evaluation: $\Ct_{a1c} = \operatorname{Re}(i^{a+c-2})$,
which equals $Z_{ac} = \delta_{ac}(3-2a)$.

\emph{Cascade rule (associativity).}\; With $M = \Ct Z$, both $(xy)z$ and
$x(yz)$ reduce, after stripping the common output $Z$, to the $Z$-dressed
chain $\sum_{d,d'} \Ct_{abd'} Z_{d'd}\, \Ct_{dce}$, which by direct evaluation
equals the fully symmetric order-4 tensor with entries
$\operatorname{Re}(i^{\,a+b+c+e-4})$; full symmetry makes the two
parenthesizations equal. The internal $Z$ is essential: the undressed chain
$\sum_d \Ct_{abd}\Ct_{dce}$ is not permutation symmetric.

\emph{Spider fusion.}\; Define the $k$-leg spider
$S_{a_1 \cdots a_k} = \operatorname{Re}\!\left(i^{\,a_1+\cdots+a_k-k}\right)$,
fully symmetric because it depends only on the index sum; $S_{ab} = Z_{ab}$
and $S_{abc} = \Ct_{abc}$. Joining one further $\Ct$ node through a
$Z$-dressed edge extends the spider by one leg,
$\sum_{d,d'} S_{a_1 \cdots a_{k-1} d}\, Z_{dd'}\, \Ct_{d'ce} =
S_{a_1 \cdots a_{k-1} ce}$, so by induction every tree of $\Ct$ nodes with
$Z$-dressed internal edges contracts to the spider on its external legs,
independent of the tree shape. Closing two legs of a spider with a
$Z$-dressed edge instead multiplies it by two,
$\sum_{d,e} S_{a_1 \cdots a_{k-2} de}\, Z_{de} = 2\, S_{a_1 \cdots a_{k-2}}$,
which is why cyclic wirings are excluded from the fusion rule.

These rules mirror the field axioms: the cascade rule is associativity; permutation
invariance of $\Ct$ is commutativity; the unit rule states $z \cdot 1 = z$ (contraction
with $\mathbf{1}$ through $\Ct$ gives $(z \cdot 1)^* = Zz$, and the trailing $Z$ of the
multiplication conjugates back to $z$); addition is componentwise (linearity); and
distributivity is the multilinearity of contraction itself.

\section{Derivation of the reverse-mode pullback}\label{app:pullback}

This appendix derives Eq.~\eqref{eq:pullback} and the wiring of
Fig.~\ref{fig:algebra}(g) from the forward contraction and the
standard multilinear chain rule.
Figure~\ref{fig:pullback-deriv} places the forward and backward
diagrams side by side with index labels for reference.

\paragraph{Forward rule (Fig.~\ref{fig:pullback-deriv}(a)).}
The realified product $C = AB$ contracts the green legs of
$T_A$ and $T_B$ through the multiplication tensor
$M_{abc} = \sum_{c'}\Ct_{abc'}\,Z_{c'c}$:
\begin{equation}\label{eq:forward-idx}
  [T_C]_{ij,c}
  = \sum_{k,\,a,\,b} [T_A]_{ik,a}\,[T_B]_{kj,b}\,M_{abc}.
\end{equation}
In the diagram, $T_A$'s green leg ($a$) and $T_B$'s green leg
($b$) enter $\Ct$ (legs~1 and~2); the third leg passes
through~$Z$ to become $T_C$'s green leg~$c$.

\paragraph{Pullback.}
The map $T_B \mapsto T_C$ is linear with $T_A$ and $M$ fixed.
For a real scalar loss~$\ell$ with upstream adjoint
$\ov{T}_C = \partial \ell/\partial T_C$,
differentiating~\eqref{eq:forward-idx} and contracting with
$\ov{T}_C$ gives
\begin{equation}\label{eq:pullback-idx}
  [\ov{T}_B]_{kj,b}
  = \sum_{i,\,a,\,c}
     [T_A]_{ik,a}\;[\ov{T}_C]_{ij,c}\;M_{abc}.
\end{equation}
This is the same $M$ as the forward rule~\eqref{eq:forward-idx},
but with the roles of~$M$'s legs swapped:
leg~2 ($b$, formerly $T_B$'s input) is now the \emph{open output}
for~$\ov{T}_B$;
leg~3 ($c$, formerly the output) is now an \emph{input}
from~$\ov{T}_C$.
The physical indices also transpose:
$T_A$ and~$\ov{T}_C$ share index~$i$
(the ``$\dagger$-direction'' contraction),
with $k$ and~$j$ open.

\paragraph{Where $Z$ lands (Fig.~\ref{fig:pullback-deriv}(b)).}
Since $Z$ is diagonal, $M_{abc} = \Ct_{abc}\,Z_c$, so
\begin{equation}\label{eq:pullback-expanded}
  [\ov{T}_B]_{kj,b}
  = \sum_{i,\,a,\,c}
    [T_A]_{ik,a}\;
    \underbrace{Z_c\,[\ov{T}_C]_{ij,c}}_{[Z\,\ov{T}_C]_{ij,c}}\;
    \Ct_{abc}.
\end{equation}
The $Z$ factor acts on $\ov{T}_C$'s green index~$c$
(conjugation); the output leg~$b$ carries no~$Z$.
This is the wiring of Fig.~\ref{fig:algebra}(g):
$T_A$'s green leg and $\ov{T}_C$'s green leg (passing
through~$Z$) enter~$\Ct$, whose third leg is
the open output for~$\ov{T}_B$.

\paragraph{Complex-picture check.}
Bare $\Ct$ computes the conjugated product:
$\sum_{a,b} x^a\,y^b\,\Ct_{abc} = [(xy)^*]_c$
(Sec.~\ref{sec:rep:local}).
In~\eqref{eq:pullback-expanded} the two inputs to $\Ct$
are $T_A$ and $Z\,\ov{T}_C = T_{\ov{C}^*}$.
The contraction gives $(A \cdot \ov{C}^*)^* = A^*\,\ov{C}$;
combined with the physical-index transpose this is
$\ov{B} = A^\dagger\,\ov{C}$.

\begin{figure}[htb]
  \centering
  \includegraphics[width=0.88\linewidth]{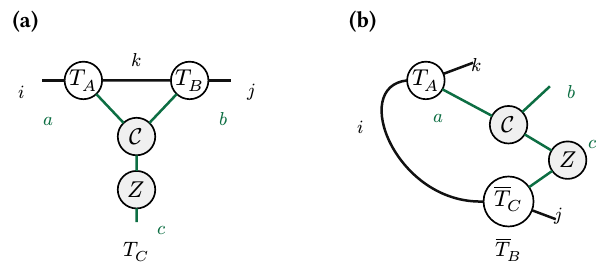}
  \caption{Forward and backward contractions with green-index labels.
  (a)~Forward product: $T_A$'s and $T_B$'s green legs ($a$,~$b$)
  enter~$\Ct$, whose third leg passes through~$Z$ to
  produce~$c$.
  (b)~Pullback with respect to~$T_B$:
  $\ov{T}_C$ enters through~$Z$ (leg~$c$, where $Z$ lives
  in~$M$), $T_A$ enters leg~$a$, and leg~$b$ is the
  open output.
  The curved black line marks the shared physical index~$i$
  (transpose).}
  \label{fig:pullback-deriv}
\end{figure}

\section{Precision gates and error budget}\label{app:precision}

The pipeline gates correctness before speed at two levels: every timed invocation first
reproduces the CPU f32 contraction of the same plan in-process, and admitted values must
then agree across backends, a check we call the \emph{pairwise output gate}. Within-device
representations, native-complex versus tree-real,
are compared at absolute tolerance~$0$ and relative tolerance~$10^{-5}$.
The within-device gate has a zero absolute floor by
design, since the amplitudes themselves can fall below $10^{-12}$.

The newer matched campaign of Sec.~\ref{sec:wallclock} adds twelve
GEMM-4M/network-3M pairs, and all twelve pass the fixed pairwise output gate. An
archived same-artifact reference provides a further check for nine of the twelve pairs;
the remaining three retain the pairwise gate only.

The gate set itself is also validated: every custom gate matrix the importer can emit is
checked for unitarity, and two exactly simulatable instances are reproduced to
${\sim}\,10^{-16}$ by an independent state-vector simulator that shares no code with the
importer. The archived plans store tensor data in f32, and on deep circuits that input
quantization contributes roughly half of the
measured error (the quantization-vs-execution split is archived with the c64
reference artifacts in the data repository); the NumPy float64 extension reference tests execution-path error but does
not remove this input-quantization contribution. The $10^{-4}$ cross-backend gate is a
conservative engineering tolerance, not an acceptance criterion tied to an application.

\bibliographystyle{unsrtnat}
\bibliography{refs}

\end{document}